\documentclass[11pt]{article}
\usepackage{fullpage}

\usepackage{mathptmx}
\usepackage{amsthm}
\usepackage{amsmath}\allowdisplaybreaks[4]
\usepackage{microtype}
\usepackage[mathcal,mathscr]{euscript} 
\usepackage{varioref}
\usepackage{nohyperref}  
\usepackage[capitalise]{cleveref}
\crefname{equation}{}{}
\usepackage{enumitem} 
\setlist{itemsep=0pt}

\newtheorem{thm}{Theorem}
\newtheorem{clm}{Claim}
\newtheorem{lem}{Lemma}

\theoremstyle{remark}
\newtheorem*{rem}{Remark}

\let\cal\mathcal
\newcommand{\A}{{\cal A}}
\newcommand{\B}{{\cal B}}
\newcommand{\R}{{\cal R}}
\newcommand{\f}{{\cal F}}
\newcommand{\minplus}{\mbox{$(\min,+)$}~}
\newcommand{\maxplus}{\mbox{$(\max,+)$}~}
\newcommand{\arithm}[1]{L(#1)}
\newcommand{\ptree}{f} 
\newcommand{\mtree}{MST}
\newcommand{\dtree}{\vec{f}}
\newcommand{\stree}{{\cal T}}
\newcommand{\dstree}{\vec{\cal T}}
\newcommand{\U}{U}
\newcommand{\V}{V}
\newcommand{\bound}[1]{\tau(#1)}
\newcommand{\Good}{\U^{\ast}}
\newcommand{\prob}[1]{\ensuremath{\mathrm{Pr}\left\{{#1}\right\}}}
\newcommand{\euler}{\mathrm{e}}
\newcommand{\redE}{E_{\mathrm{red}}}

\begin{document}
\title{Greedy Can Beat Pure Dynamic Programming\thanks{Research supported by the DFG grant
    JU~3105/1-1 (German Research Foundation).}}

\author{Stasys Jukna\thanks{Affiliated
    with the Institute of Data Science and Digital Technologies, Faculty of Mathematics and Informatics of Vilnius University, Lithuania. E-mail: {\sf stjukna@gmail.com}} \and  Hannes Seiwert\thanks{E-mail: {\sf seiwert@thi.cs.uni-frankfurt.de} }}
\date{\footnotesize\it  Institut f\"ur Informatik, Goethe Universit\"at\\[-0.3em]
    Frankfurt am Main, Germany
}

\maketitle

\begin{abstract}
  Many dynamic programming algorithms for discrete 0-1 optimization
  problems are ``pure'' in that their recursion equations only use
  min/max and addition operations, and do not depend on actual input
  weights.  The well-known greedy algorithm of Kruskal solves the
  minimum weight spanning tree problem on $n$-vertex graphs using only
  $O(n^2\log n)$ operations. We prove that any pure DP algorithm for
  this problem must perform $2^{\Omega(\sqrt{n})}$ operations. Since
  the greedy algorithm can also badly fail on some optimization
  problems, easily solvable by pure DP algorithms, our result shows
  that the computational powers of these two types of algorithms are
  incomparable.\\

\end{abstract}

\section{Introduction and result}

A dynamic programming (DP) algorithm is \emph{pure} if it only uses
min or max and addition operations in its recursion equations, and the
equations do not depend on the actual values of the input weights.
Notable examples of such DP algorithms include the
Bellman--Ford--Moore algorithm for the shortest $s$-$t$ path
problem~\cite{ford,Moore1957,bellman}, the Floyd--Warshall algorithm
for the all-pairs shortest paths problem~\cite{floyd,warshall}, the
Held--Karp algorithm for the Travelling Salesman
Problem~\cite{held62}, and the Dreyfus--Levin--Wagner algorithm for
the weighted Steiner tree problem~\cite{dreyfus,levin}.

It is well known and easy to show that, for some optimization
problems, already pure DP algorithms can be much better than greedy
algorithms.  Namely, there are a lot of optimization problems which
are easily solvable by pure DP algorithms (exactly), but the greedy
algorithm cannot even achieve any finite approximation factor: maximum
weight independent set in a path, or in a tree, the maximum weight
simple $s$-$t$ path in a transitive tournament problem, etc.

In this paper, we show that the converse direction also holds: for
some optimization problems, greedy algorithms can also be much better
than pure dynamic programming. So, the computational powers of greedy
and pure DP algorithms are \emph{incomparable}. We will show that the
gap occurs on the (undirected) minimum weight spanning tree problem,
by first deriving an exponential lower bound on the monotone
arithmetic circuit complexity of the corresponding polynomial.

Let $K_n$ be the complete undirected graph on
$[n]=\{1,\ldots,n\}$. Assume that edges $e$ have their associated
nonnegative real weights $x_{e}$, considered as formal variables. Let
$\stree_n$ be the family of all $|\stree_n|=n^{n-2}$ spanning trees
$T$ in $K_n$, each viewed as its \emph{set} of edges.

It is well known that $\stree_n$ is the family of bases of a matroid,
known as \emph{graphic matroid}; so, on this family of feasible
solutions, both optimization problems (minimization and maximization)
can be solved by standard greedy algorithms. On the other hand, the
theorem below states that the polynomial corresponding to $\stree_n$
has exponential monotone arithmetic circuit complexity; due to special
properties of this polynomial, the same lower bound also holds on the
number of operations used by pure DP algorithms solving minimization
and maximization problems on $\stree_n$ (see~\cref{lem:reduct}).

The \emph{spanning tree polynomial} (known also as the \emph{Kirchhoff
  polynomial} of $K_n$) is the following homogeneous, multilinear
polynomial of degree $n-1$:
\[
\ptree_n(x) = \sum_{T\in\stree_n} \prod_{e\in T}x_e\,.
\]
For a multivariate polynomial $f$ with positive coefficients, let
$\arithm{f}$ denote the minimum size of a monotone arithmetic
$(+,\times)$ circuit computing~$f$. Our goal is to prove that
$\arithm{\ptree_n}$ is exponential in~$n$.

\begin{thm}\label{thm:main}
  \[
  \arithm{\ptree_n}=2^{\Omega(\sqrt{n})}\,.
  \]
\end{thm}

A ``directed version'' of $\ptree_n$ is the \emph{arborescence
  polynomial}~$\dtree_n$.  An \emph{arborescence} (known also as a
\emph{branching} or a \emph{directed spanning tree}) on the vertex-set
$[n]$ is a directed tree with edges oriented away from vertex $1$ such
that every other vertex is reachable from vertex $1$.  Let $\dstree_n$
be the family of all arborescences on $[n]$. Jerrum and
Snir~\cite{jerrum} have shown that $\arithm{\dtree_n}=2^{\Omega(n)}$
holds for the arborescence polynomial
\[
\dtree_n(x) = \sum_{T\in\dstree_n}\ \prod_{ \vec{e}\in
  T}x_{\vec{e}}\,.
\]
Note that here variables $x_{i,j}$ and $x_{j,i}$ are treated as
\emph{distinct}, and cannot both appear in the same monomial.  This
dependence on orientation was crucially utilized in the argument
of~\cite[p.~892]{jerrum} to reduce a trivial upper bound $(n-1)^{n-1}$
on the number of monomials in a polynomial computed at a particular
gate till a non-trivial upper bound $(3n/4)^{n-1}$.  So, this argument
does not apply to the \emph{undirected} version~$\ptree_n$ (where
$x_{i,j}$ and $x_{j,i}$ stand for the \emph{same} variable). To handle
the undirected case, we will use an entirely different argument.

\paragraph{Relation to pure DP algorithms}

Every pure DP algorithm is just a special (recursively constructed)
\emph{tropical} \minplus or \maxplus circuit, that is, a circuit using
only min (or max) and addition operations as gates; each input gate of
such a circuit holds either one of the variables $x_i$ or a
nonnegative real number.  So, lower bounds on the size of tropical
circuits yield the same lower bounds on the number of operations used
by pure DP algorithms. For optimization problems, whose feasible
solutions all have the same cardinality, the task of proving lower
bounds on their tropical circuit complexity can be solved by proving
lower bounds on the size of monotone arithmetic circuits.

Recall that a multivariate polynomial is \emph{monic} if all its
nonzero coefficients are equal to $1$, \emph{multilinear} if no
variable occurs with degree larger than $1$, and \emph{homogeneous} if
all monomials have the same degree.  Every monic and multilinear
polynomial $f(x)=\sum_{S\in\f}\prod_{i\in S} x_i$ defines two
optimization problems: compute the minimum or the maximum of
$\sum_{i\in S}x_i$ over all $S\in\f$.

\begin{lem}[\cite{jerrum,juk14}]\label{lem:reduct}
  If a polynomial $f$ is monic, multilinear and homogeneous, then
  every tropical circuit solving the corresponding optimization
  problem defined by $f$ must have at least $\arithm{f}$ gates.
\end{lem}

This fact was proved by Jerrum and Snir~\cite[Corollary~2.10]{jerrum};
see also \cite[Theorem~9]{juk14} for a simpler proof. The proof idea
is fairly simple: having a tropical circuit, turn it into a monotone
arithmetic $(+,\times)$ circuit, and use the homogeneity of $f$ to
show that, after removing some of the edges entering $+$-gates, the
resulting circuit will compute our polynomial~$f$.

\paragraph{Greedy can beat pure DP}
The (weighted) \emph{minimum spanning tree} problem $\mtree_n(x)$ is,
given an assignment of nonnegative real weights to the edges of $K_n$,
compute the minimum weight of a spanning tree of $K_n$, where the
weight of a graph is the sum of weights of its edges. So, this is
exactly the minimization problem defined by the spanning tree
polynomial~$\ptree_n$:
\[
\mtree_n(x) = \min_{T\in\stree_n} \ \sum_{e\in T}x_e\,.
\]
Since the family $\stree_n$ of feasible solutions of this problem is
the family of bases of the (graphic) matroid, the problem \emph{can}
be solved by the standard greedy algorithm. In particular, the
well-known greedy algorithm of Kruskal~\cite{kruskal} solves $\mtree_n$
using only $O(n^2\log n)$ operations.

On the other hand, since the spanning tree polynomial $\ptree_n$ is
monic, multilinear and homogeneous, \cref{thm:main} together with
\cref{lem:reduct} implies that any \minplus circuit solving the
problem $\mtree_n$ must have at least
$\arithm{\ptree_n}=2^{\Omega(\sqrt{n})}$ gates and, hence, at least so
many operations must be performed by any pure DP algorithm solving
$\mtree_n$. This gap between pure DP and greedy algorithms is our main
result.

\paragraph{Directed versus undirected spanning trees}
  The arborescence polynomial~$\dtree_n$ is also monic, multilinear
  and homogeneous, so that \cref{lem:reduct}, together with the above
  mentioned lower bound on $\arithm{\dtree_n}$ due to Jerrum and
  Snir~\cite{jerrum}, also yields the same lower bound on the size of
  \minplus circuits solving the minimization problem on the family
  $\dstree_n$ of arborescences.

  But this does not separate DP from greedy, because the downward
  closure of $\dstree_n$ is \emph{not} a matroid: it is only an
  intersection of two matroids (see Edmonds~\cite{edmonds67}). So,
  greedy algorithms are only able to \emph{approximate} the
  minimization problem on $\dstree_n$ within the factor $2$.
  Polynomial time algorithms solving this problem \emph{exactly} were
  found by several authors, starting from Edmonds~\cite{edmonds73}.
  The fastest algorithm for the problem is due to
  Tarjan~\cite{tarjan}, and solves it in time $O(n^2\log n)$,
  that is, with the same time complexity as Kruskal's greedy algorithm
  for undirected graphs~\cite{kruskal}.  But these are not greedy
  algorithms. So, $\dstree_n$ does not separate standard, matroid based greedy and pure DP
  algorithms.

\section{Proof of Theorem~\ref{thm:main}}

A \emph{rectangle} is specified by giving two families $\A$ and $\B$
of forests in the complete graph $K_n$ on $[n]=\{1,\ldots,n\}$ such
that for all forests $A\in\A$ and $B\in\B$ (viewed as sets of their
edges), we have $A\cap B=\emptyset$ (the forests are edge-disjoint),
and $A\cup B$ is a spanning tree of $K_n$.  The rectangle itself is
the family
\[
\R=\A\lor\B:=\{A\cup B\colon \mbox{$A\in\A$ and $B\in\B$}\}
\]
of all resulting spanning trees. A rectangle $\R=\A\lor\B$ is
\emph{balanced} if $(n-1)/3\leq |A|,|B|\leq 2(n-1)/3$ holds for all
forests $A\in\A$ and $B\in\B$; recall that every spanning tree of a
graph on $n$ vertices has $n-1$ edges.  Let $\bound{n}$ be the minimum
number of balanced rectangles whose union gives the family of all
spanning trees of~$K_n$.

\begin{lem}\label{lem:tau}
  For the spanning tree polynomial $\ptree_n$, we have
  $\arithm{\ptree_n}\geq \bound{n}$.
\end{lem}

\begin{proof}
  Let $t=\arithm{\ptree_n}$. The spanning tree polynomial $\ptree_n$
  is multilinear and homogeneous of degree $n-1$: every spanning tree
  $T$ of $K_n$ has $|T|=n-1$ edges. Since the polynomial $\ptree_n$ is
  homogeneous of degree $n-1$, and since $\ptree_n$ can be computed by
  a monotone arithmetic circuit of size $t$, the well-known
  decomposition result, proved by Hyafil~\cite[Theorem~1]{hyafil} and
  Valiant~\cite[Lemma~3]{valiant80}, implies that $\ptree_n$ can be
  written as a sum $\ptree_n=g_1\cdot h_1+\cdots+g_t\cdot h_t$ of
  products of nonnegative homogeneous polynomials, each of degree at
  most $2(n-1)/3$; a polynomial is \emph{nonnegative} if it has no
  negative coefficients.

  Every monomial of $\ptree_n$ is of the form $\prod_{e\in T}x_e$ for
  some spanning tree~$T$. Since the polynomials $g_i$ and $h_i$ in the
  decomposition of $\ptree_n$ are nonnegative, there can be no
  cancellations. This implies that all the monomials of $g_i\cdot h_i$
  must be also monomials of $\ptree_n$, that is, must correspond to
  spanning trees. Moreover, since the polynomial $\ptree_n$ is
  multilinear, the forests of $g_i$ must be edge-disjoint from the
  forests of $h_i$. So, if we let $\A_i$ be the family of forests
  corresponding to monomials of the polynomial $g_i$, and $\B_i$ be
  the family of forests corresponding to monomials of the
  polynomial~$h_i$, then $\A_1\lor\B_1,\ldots,\A_t\lor\B_t$ are
  balanced rectangles, and their union gives the family of all spanning trees
  of $K_n$. This shows $\bound{n}\leq t=\arithm{\ptree_n}$, as
  desired.
\end{proof}

So, it is enough to prove an exponential lower bound on $\bound{n}$.
When doing this, we will concentrate on spanning trees of $K_n$ of a
special form.  Let $m$ and $d$ be positive integer parameters
satisfying $(d+1)m=n$, $m=\Theta(\sqrt{n})$ and $m\leq d/32$; we will
specify these parameters later.

A \emph{star} is a tree with one vertex, the \emph{center}, adjacent
to all the others, which are \emph{leaves}. A $d$-\emph{star} is a
star with $d$ leaves.  A \emph{spanning star-tree} consists of $m$
vertex-disjoint $d$-stars whose centers are joined by a path.  A
\emph{star factor} is a spanning forest of $K_n$ consisting of $m$
vertex-disjoint $d$-stars. Note that each spanning star-tree contains
a unique star factor (obtained by removing edges between star
centers).

Let $\f$ be the family of all star factors of $K_n$. For a rectangle
$\R$, let $\f_{\R}$ denote the family of all star factors $F$ of $K_n$
contained in at least one spanning tree of $\R$; in this case, we also
say that the factor $F$ is \emph{covered} by the rectangle~$\R$.

\begin{lem}\label{lem:main}
  There is an absolute constant $c>0$ such that for every balanced
  rectangle $\R$, we have $|\f_{\R}|\leq |\f|\cdot 2^{-c\sqrt{n}}$.
\end{lem}

Note that this lemma gives a lower bound $\bound{n}\geq 2^{c\sqrt{n}}$
on the minimum number of balanced rectangles containing all spanning
trees of $K_n$. Indeed, let $\R_1,\ldots,\R_t$ be $t=\bound{n}$
balanced rectangles whose union is the family of all spanning trees of
$K_n$. Every star factor $F\in\f$ is contained in at least one
spanning tree (in fact, in many of them). So, every star factor
$F\in\f$ must be covered by at least one of these $t$ rectangles. But
\cref{lem:main} implies that none of these rectangles can cover more
than $h:=|\f|\cdot 2^{-c\sqrt{n}}$ star factors $F\in\f$. So, we need
$\bound{n}=t\geq |\f|/h\geq 2^{c\sqrt{n}}$ rectangles.  Together with
\cref{lem:tau}, this yields the desired lower bound
$\arithm{\ptree_n}\geq 2^{c\sqrt{n}}$ on the monotone arithmetic
circuit complexity of the spanning tree polynomial~$\ptree_n$.

The rest of the paper is devoted to the proof of~\cref{lem:main}.

\begin{proof}[Proof of~\cref{lem:main}]
  We can construct every star factor $F\in\f$ using the following
  procedure.
  \begin{enumerate}
  \item Choose a subset of $m$ centers in $[n]$; $\binom{n}{m}$
    possibilities.
  \item Divide the remaining $n-m$ vertices into $m$ blocks of size
    $d$, and connect all vertices of the $i$th block to the $i$th
    largest of the chosen centers; there are
    $\binom{n-m}{d,\ldots,d}=\frac{(n-m)!}{d!^m}$ possibilities to do
    this.
  \end{enumerate}
  Since different realizations of this procedure lead to different
  star factors, we have
  \begin{equation}\label{eq:all}
    |\f|=\binom{n}{m}\frac{(n-m)!}{d!^m}\,.
  \end{equation}

  Fix a balanced rectangle $\R=\A\lor\B$ containing at least one
  spanning star-tree $T_0=A_0\cup B_0$ with $A_0\in\A$ and $B_0\in\B$, and let $c_1,\ldots,c_m$ be
  the centers of stars of $T_0$. Every vertex $v\in
  [n]\setminus\{c_1,\ldots,c_m\}$ is connected in $T_0$ by a
  \emph{unique} edge $e_v$ to one of the centers
  $c_1,\ldots,c_m$. This gives us a partition $\U\cup \V$ of the
  vertices in $[n]\setminus\{c_1,\ldots,c_m\}$ into two sets
  determined by the forests $A_0$ and $B_0$:
  \[
  \U=\{v\colon e_v\in A_0\}\ \ \mbox{ and }\ \ \V=\{v\colon e_v\in
  B_0\}\,.
  \]
  We will concentrate on the bipartite complete subgraph $\U\times\V$
  of $K_n$, and call the edges of $K_n$ lying in this subgraph
  \emph{crossing edges}.  Since our rectangle $\R$ is balanced, we
  know that both $|A_0|$ and $|B_0|$ lie between $(n-1)/3$ and
  $2(n-1)/3$. So, since $m=o(n)$, for $n$ large enough, we have
  \begin{equation}\label{eq:balance}
    |\U|,|\V|\geq \tfrac{1}{3}(n-1)-m \geq \tfrac{1}{4}n\,.
  \end{equation}

  The property that every graph $A\cup B$ with $A\in\A$ and $B\in\B$
  must be cycle-free (must be a spanning tree of $K_n$) gives the
  following restriction on the rectangle $\R=\A\lor\B$.

\begin{clm}\label{clm:comb}
  For all forests $A\in\A$ and $B\in\B$, and vertices $u\in\U$ and
  $v\in\V$, we have $|A\cap(\{u\}\times\V)|\leq m$ and
  $|B\cap(\U\times \{v\})|\leq m$.
\end{clm}
That is, no forest $A\in\A$ can contain more than $m$ crossing edges
incident to one vertex in $U$, and no forest $B\in\B$ can contain more
than $m$ crossing edges incident to one vertex in~$V$.

\begin{proof}
  Assume contrariwise that some vertex $u\in\U$ has $l\geq m+1$
  crossing edges $\{u,v_1\},\ldots,\{u,v_l\}$ in the forest $A$.
  Since these edges are crossing and $u\in\U$, all vertices
  $v_1,\ldots,v_l$ belong to $\V$. In the (fixed) spanning
  star-tree $T_0=A_0\cup B_0$ (determining the partition $\U\cup\V$ of
  vertices in $[n]\setminus\{c_1,\ldots,c_m\}$) each of these $l$ vertices is joined by an edge of
  the forest $B_0$ to one of the centers $c_1,\ldots,c_m$ of stars of
  $T_0$.

  Since $l > m$, some two of these vertices $v_i$ and $v_j$ must be
  joined in $B_0$ to the same center $c\in\{c_1,\ldots,c_m\}$. Since
  $\R$ is a rectangle, the graph $A\cup B_0$ must be a (spanning)
  tree. But the edges $\{u,v_i\}, \{u,v_j\}$ of $A$ together with
  edges $\{v_i,c\}, \{v_j,c\}$ of $B_0$ form a cycle $u\to v_i\to c\to
  v_j\to u$ in $A\cup B_0$, a contradiction.

  The proof of the inequality $|B\cap(\U\times \{v\})|\leq m$ is the
  same by using the forest $A_0$ instead of~$B_0$.
\end{proof}

So far, we only used one fixed spanning tree $T_0$ in the rectangle
$\R$ to define the subgraph $\U\times\V$ of $K_n$. We now use the
entire rectangle $\R=\A\lor\B$ to color the \emph{edges} of $K_n$ in
red and blue. When doing this, we use the fact that the sets
$E_{\A}:=\bigcup_{A\in\A} A$ and $E_{\B}:=\bigcup_{B\in\B} B$ of edges
of $K_n$ must be disjoint:
\begin{itemize}
\item Color an edge $e\in K_n$ \emph{red} if $e\in E_{\A}$, and color $e$
  \emph{blue} if $e\in E_{\B}$.
\end{itemize}
This way, the edges of every spanning tree $T\in\R$ will receive their
colors. The remaining edges of $K_n$ (if there are any) can be colored
arbitrarily.

Recall that an edge $e$ of $K_n$ is \emph{crossing} if
$e\in\U\times\V$. Assume that at least half of the crossing edges is
colored in \emph{red}; otherwise, we can consider blue edges. This
assumption implies that the set $\redE\subseteq \U\times\V$ of red
crossing edges has $|\redE|\geq \tfrac{1}{2}|\U\times\V|$ edges.  For
a vertex $u\in U$, the set of its \emph{good neighbors} is the set
\[
\V_u=\left\{v\in\V\colon \{u,v\}\in \redE\right\}
\]
of vertices that are connected to $u$ by red crossing edges.
\Cref{clm:comb} gives the following structural restriction on star
factors covered by the rectangle~$\R$.

\begin{clm}\label{clm:comb1}
For any star factor $F\in\f_{\R}$, and for any center $z$ of $F$, if $z\in\U$, then $|F\cap(\{z\}\times\V_z)|\leq m$.
\end{clm}

That is, if a star factor $F$ is covered by the rectangle $\R$, then
every star of $F$ centered in some vertex $z\in \U$ can only have $m$
or fewer (out of all $|\V_z|$ possible) red crossing edges.

\begin{proof}
  Take a star factor $F\in\f_{\R}$ having some star whose center $z$
  belongs to~$\U$. Since $F$ is covered by the rectangle $\R$, $F\subseteq
  A\cup B$ holds for some forests $A\in\A$ and $B\in\B$. By the
  definition of the edge-coloring, we have
  $B\cap(\{z\}\times\V_z)=\emptyset$: all edges in $\{z\}\times\V_z$
  are red, while those in $B$ are blue.  So, all edges of
  $F\cap(\{z\}\times\V_z)$ belong to the forest $A$, and
  \cref{clm:comb} yields $|F\cap(\{z\}\times\V_z)|\leq
  |A\cap(\{z\}\times\V_z)|\leq m$.
\end{proof}

We call a vertex $u$ of $K_n$ \emph{rich} if $u\in\U$ and at least one
quarter of the vertices in $\V$ are good neighbors of $u$, that is, if
$|\V_u| \geq \tfrac{1}{4}|\V|$ holds. By \cref{eq:balance}, every rich
vertex $u$ has $|\V_u|\geq n/16$ good neighbors.  Split the family
$\f_{\R}$ of star factors covered by the rectangle $\R$ into the
family $\f^1_{\R}$ of star factors $F\in\f_{\R}$ with \emph{no} rich
center, and the family $\f^2_{\R}$ of all star factors $F\in\f_{\R}$
with \emph{at least} one rich center. We will upper-bound the number
of star factors in $\f^1_{\R}$ and in $\f^2_{\R}$ separately.

The intuition behind this splitting is that star factors
$F\in\f^1_{\R}$ have the restriction (given by \cref{lem:rich} below)
that only relatively ``few'' potential vertices of $K_n$ can be used
as centers of stars, while the restriction for the star factors
$F\in\f^2_{\R}$ (given by \cref{clm:comb1}) is that at least one of
its stars $S_z\subset F$ (centered in a rich center $z$) has
relatively ``few'' potential vertices of $K_n$ which can be taken as
leaves.

To upper-bound $|\f^1_{\R}|$, let us first show that the set
$\Good=\left\{u\in\U\colon |\V_u| \geq \tfrac{1}{4}|\V| \right\}$ of
all rich vertices is large enough.

\begin{clm}\label{lem:rich}
  There are $|\Good|\geq \tfrac{1}{4}|\U|\geq n/16$ rich vertices.
\end{clm}
\begin{proof}
  The second inequality follows from~\cref{eq:balance}.  To prove the
  first inequality, assume contrariwise that there are only $|\Good| <
  \tfrac{1}{4}|\U|$ rich vertices in $\U$. Since $|\V_u| <
  \tfrac{1}{4}|\V|$ holds for every vertex $u \in \U \setminus \Good$,
  we obtain
  \begin{align*}
    \tfrac{1}{2}|\U\times\V|&\leq |\redE| = \sum_{u\in\Good}|\V_u|
    + \sum_{u\in\U\setminus\Good}|\V_u| < \tfrac{1}{4}|\U| \cdot
    |\V|+|\U| \cdot \tfrac{1}{4}|\V| = \tfrac{1}{2}|\U \times \V|\,,
  \end{align*}
  a contradiction.
\end{proof}

Each star factor in $\f^1_{\R}$ can be constructed in the same way as
we constructed any star factor $F \in \f$ above (before
\cref{eq:all}), with the difference that centers can only be chosen
from $[n] \setminus \Good$, not from the entire set $[n]$.  Thus,
\begin{equation}\label{eq:1}
  \frac{|\f^1_{\R}|}{|\f|} \leq \binom{n-|\Good|}{m}\cdot \binom{n}{m}^{-1} \leq \euler^{-|\Good| \cdot m/n }=2^{-\Omega(m)}\,.
\end{equation}
Here we used \cref{lem:rich} together with the second of the two
simple inequalities holding for all $b\leq b+x < a$:
\begin{equation}\label{eq:bollobas}
  \left(\frac{a-b-x}{a-x}\right)^x\leq \binom{a-x}{b}\binom{a}{b}^{-1}
  \leq \left(\frac{a-b}{a}\right)^x\,.
\end{equation}

To upper bound $|\f^2_{\R}|$, we will use the restriction given by
\cref{clm:comb}. Recall that every star factor $F\in\f^2_{\R}$ has at
least one rich center.  So, consider the following (nondeterministic)
procedure of constructing a star factor $F$ in $\f^2_{\R}$.

\begin{enumerate}
\item Choose a rich center $z\in\Good$; there are at most $|\Good|\leq
  |\U|\leq n$ possibilities to do this.

\item For the center $z$, do the following:
  \begin{enumerate}
  \item choose a subset of $i\leq m$ vertices from the set $\V_z$ of
    all good neighbors of $z$, and connect these vertices to $z$ by
    (crossing) edges; for each $i\leq m$ there are $\binom{|\V_z|}{i}$
    possibilities.
  \item choose a subset of $d-i$ vertices in $[n]\setminus
    (\V_z\cup\{z\})$ and connect them to $z$; here we have at most
    $\binom{n-|\V_z|-1}{d-i}\leq \binom{n-|\V_z|}{d-i}$ possibilities.
  \end{enumerate}
\item Choose a subset of $m-1$ distinct centers from the remaining
  $n-d-1$ vertices. There are at most $\binom{n-d-1}{m-1}\leq
  \binom{n-1}{m-1}=\frac{m}{n}\binom{n}{m}$ possibilities to do this.

\item Choose a partition of the remaining $n-m-d$ vertices into $m-1$
  blocks of size~$d$, and connect the $i$th largest of the $m-1$
  chosen centers to all vertices in the $i$th block. There are at most
  $\binom{n-m-d}{d,\ldots,d}=\frac{(n-m-d)!}{d!^{m-1}}$ possibilities
  to do this.
\end{enumerate}

\begin{clm}\label{clm:procedure}
  Every star factor $F\in\f^2_{\R}$ can be produced by the above
  procedure.
\end{clm}

\begin{proof}
  Take a star factor $F\in\f_{\R}$ containing a star $S_z\subset F$
  centered in a rich vertex $z\in\Good$. The star $z$ can be picked
  by Step~1 of the procedure. By \cref{clm:comb1}, the star $S_z$ can
  only have $i:=|F\cap(\{z\}\times \V_z)|\leq m$ good neighbors of $z$ (those
  in $\V_z$) as leaves, and Step~2(a) of our procedure can pick all
  these $i$ leaves of $S_z$.  The remaining $d-i$ leaves of the star
  $S_z$ must belong to the set $[n]\setminus(\V_z\cup\{z\})$. So,
  Step~2(b) can pick these $d-i$ leaves of~$S_z$. Since the remaining
  two steps 3 and 4 of the procedure can construct \emph{any} star factor
  of $K_n\setminus S_z$, the rest of the star factor $F$ can be
  constructed by these steps.
\end{proof}

The number of possibilities in Step~2 of our procedure is related to
the probability distribution
\[
h(K,n,d,i):=\prob{X=i}=\frac{\binom{K}{i}\binom{n-K}{d-i}}{\binom{n}{d}}
\]
of a hypergeometric random variable $X$: the probability of having
drawn exactly $i$ white balls, when drawing uniformly at random
without replacement $d$ times, from a vase containing $K$ white and
$n-K$ black balls.  The number of possibilities in Step~2 of the
procedure (for a center $z$ picked in Step~1) is then at most $
H(|\V_z|,n,d,m)\cdot \binom{n}{d}$, where
\[
H(K,n,d,m):=\prob{X\leq m}=\sum_{i=0}^m h(K,n,d,i)\,,
\]
is the probability of having drawn at most $m$ white balls.  For fixed
$n,d$ and $m$, the function $H(K,n,d,m)$ is non-increasing in $K$,
implying that the maximum of $H(|\V_z|,n,d,m)$ over all rich centers
$z\in\Good$ is achieved for $K:=\min\{|\V_z|\colon
z\in\Good\}$. Hence, for every rich center $z\in\Good$, the number of
possibilities in Step~2 is at most
\[
H(|\V_z|,n,d,m)\cdot \binom{n}{d}\leq H\cdot \binom{n}{d}\,,
\]
where $H:=H(K,n,d,m)$. From the first inequality of \cref{eq:bollobas}
(applied with $x:=m$, $a:=n$ and $b:=d$) we have $\binom{n}{d} \leq
C\cdot \binom{n-m}{d}$, where $C=\left(\frac{n-m}{n-d-m}\right)^m \leq
\exp\left(\tfrac{md}{n-d-m}\right)$ is a constant since $md=O(n)$ and
$m,d=o(n)$.

Thus the total number of possibilities in all steps 1--4 and, by
\cref{clm:procedure}, also the number $|\f^2_{\R}|$ of star factors in
$\f^2_{\R}$, is at most a constant times
\begin{align*}
  &\underbrace{n\cdot H\cdot \binom{n-m}{d}}_{\text{Steps 1 and 2}}
  \underbrace{\frac{m}{n} \binom{n}{m}}_{\text{Step 3}}
  \underbrace{\frac{(n-m-d)!}{d!^{m-1}}}_{\text{Step 4}} =m\cdot
  H\cdot \underbrace{\binom{n}{m}\frac{(n-m)!}{d!^m}}_{=\ |\f|} \, .
\end{align*}
Known tail inequalities for the hypergeometric distribution (see
Hoeffding~\cite{hoeffding}, or Chv\'atal~\cite{chvatal} for a direct
proof) imply that, if $m\leq (K/n-\epsilon)d$ for $\epsilon >0$, then
\begin{equation}\label{eq:tail}
  H(K,n,d,m)= \prob{X\leq m}
  \leq \euler^{-2\epsilon^2 d}\,.
\end{equation}
\begin{rem}
  In both papers \cite{hoeffding} and \cite{chvatal}, this upper bound
  is only stated for the event $X\geq (K/n+\epsilon)d$, but using the
  duality $h(K,n,d,i)=h(n-K,n,d,d-i)$ (count black balls instead of
  white), the same upper bound holds also for the event $X\leq
  (K/n-\epsilon)d$.
\end{rem}

In our case, $K=\min\{|\V_z|\colon z\in\Good\}\geq
\tfrac{1}{4}|\V|\geq n/16$.  Recall that, so far, we have only used
the conditions $(d+1)m=n$ and $m=\Theta(\sqrt{n})$ on the parameters
$m$ and $d$. We now use the last condition $m\leq d/32$. For $\epsilon
=1/32$, we then have $m\leq d/32\leq (K/n-\epsilon)d$, and
\cref{eq:tail} yields
\[
H=H(K,n,d,m)\leq \prob{X\leq d/32}\leq
\euler^{-d/512}=2^{-\Omega(d)}\,.
\]
By taking $d:=6\sqrt{n}$ and $m:=n/(d+1)$, all three conditions on the
parameters $m$ and $d$ are fulfilled, and we obtain $|\f^2_{\R}|\leq
m|\f|\cdot 2^{-\Omega(d)}$.  Together with
the upper bound \cref{eq:1}, the desired upper bound on $|\f_{\R}|$
follows:
\[
\frac{|\f_{\R}|}{|\f|}=\frac{|\f^1_{\R}|+|\f^2_{\R}|}{|\f|}\leq
2^{-\Omega(m)}+m2^{-\Omega(d)}=2^{-\Omega(\sqrt{n})}\,. \qedhere
\]
\end{proof}

\small

\end{document}